\documentclass[runningheads]{llncs}
\usepackage[T1]{fontenc}
\usepackage{graphics,amsmath,amssymb,arydshln,yhmath,pifont}
\usepackage{tikz}
\usetikzlibrary{positioning,arrows.meta,calc}
\usepackage{mathtools}
\mathtoolsset{showonlyrefs}
\usepackage[hidelinks]{hyperref}

\newtheorem{Proposition}{Proposition}
\newtheorem{Definition}{Definition}
\newtheorem{Assumption}{Assumption}
\newtheorem{Lemma}{Lemma}

\def\Ac{{\mathcal A}}

\def\Nc{{\mathcal N}}

\def\Rbb{{\mathbb R}}

\def\Sc{{\mathcal S}}

\def\0{{\bf 0}}

\newcommand{\bitem}{\begin{itemize}}
\newcommand{\eitem}{\end{itemize}}
\newcommand{\btabular}{\begin{tabular}}
\newcommand{\etabular}{\end{tabular}}
\newcommand{\bcenter}{\begin{center}}
\newcommand{\ecenter}{\end{center}}
\newcommand{\bea}{\begin{eqnarray}}
\newcommand{\eea}{\end{eqnarray}}
\newcommand{\bean}{\begin{eqnarray*}}
\newcommand{\eean}{\end{eqnarray*}}
\newcommand{\ba}{\left[ \begin{array}}
\newcommand{\ea}{\\ \end{array} \right]}
\newcommand{\bear}{\begin{array}}
\newcommand{\eear}{\\ \end{array}}

\newcommand{\non}{\nonumber}

\newcommand{\ra}{\rightarrow}

\newcommand*{\QEDB}{\hfill\ensuremath{\blacksquare}}%
\newcommand*{\QET}{\hfill\ensuremath{\triangleleft}}
\newcommand{\norm}[1]{\left\lVert#1\right\rVert}

\newcounter{subequation}
\def\beasub{\addtocounter{equation}{+1}
\setcounter{subequation}{\value{equation}}
\setcounter{equation}{0}
\renewcommand{\theequation}{\arabic{subequation}\alph{equation}}
\begin{eqnarray}}
\def\eeasub{\end{eqnarray}
\setcounter{equation}{\value{subequation}}
\renewcommand{\theequation}{\arabic{equation}}}

\usepackage{color}

\begin{document}
\title{Divergence-based Safety Measure for Large Language Models via Rational Inattention}
\titlerunning{Divergence-based Safety Measure for LLMs}
\author{Anh Tung Nguyen\inst{1}\orcidID{0000-0001-9316-233X} \and
\\
Quanyan Zhu\inst{2}\orcidID{0000-0002-0008-2953}
}
\authorrunning{A. T. Nguyen and Q. Zhu}
\institute{Department of Information Technology, Uppsala University, Box 337, SE 75105, Sweden
\\
\email{anh.tung.nguyen@it.uu.se}
\and
Department of Electrical and Computer Engineering,
New York University, NY, 11201, USA
\\
\email{qz494@nyu.edu}
}
\maketitle              % typeset the header of the contribution
\begin{abstract}
This paper proposes a divergence-based safety measure for large language models (LLMs) under embedding-input attacks. The proposed measure quantifies the worst-case Kullback--Leibler divergence between the clean and attacked LLMs' output distributions, subject to a stealthiness constraint. 
This constraint is constructed by leveraging the equivalence between transformer attention used in LLMs and rational inattention modeling human decision-making.
We analyze the proposed divergence-based safety measure by investigating perfectly undetectable attacks and deriving its upper bound through a Bregman-divergence argument. 
The proposed safety measure is applied to two pretrained causal language models, GPT-2 and GPT-Neo-125M, to show nontrivial output-distribution shifts, illustrating that the measure can distinguish model-level safety profiles.
\end{abstract}
\section{Introduction}
%% New introduction

Large language models (LLMs), such as OpenAI GPT and Meta Llama, have been widely employed in many real-world applications thanks to their ability to generate text responses to human queries \cite{minaee2024large}. As LLMs are increasingly integrated into decision-support pipelines, human operators may in the near future consult LLM-generated outputs before making consequential decisions. This deployment paradigm introduces a system-level safety concern, where malicious perturbations of LLMs' inputs may alter the model's outputs and consequently influence the downstream decision made by a human operator.

This concern motivates the need for quantitative safety measures for LLMs. 
More specifically, from a systems and control perspective, a measure quantifying the worst-case output perturbation caused by input attacks is surely demanded, e.g., \cite{nguyen2025scalable}. Such a safety measure can support model comparison, runtime monitoring, and fine-tuning of LLM parameters to mitigate the impact of cyber-attacks.
However, to the best of our knowledge, limited progress has been made from a systems and control perspective toward developing such safety measures for LLMs \cite{chang2024survey}.
% no such safety measure for LLMs has been introduced with theoretical support. 

To fill this gap, this paper introduces a divergence-based safety measure for LLMs. We consider attacks on the embedding-input tokens and measure the attack impact by the Kullback--Leibler (KL) divergence between the clean and attacked LLMs' output distributions. The proposed measure quantifies the worst-case output-distribution divergence induced by energy-bounded attacks that remain stealthy with respect to an internal monitoring output. A larger value of the proposed measure indicates that the LLMs' output distribution can be changed more significantly by a stealthy attack. Hence, a smaller value corresponds to a safer model in the sense of reduced worst-case stealthy output impact.

The internal monitoring output is constructed from transformer attention \cite{vaswani2017attention}, commonly used in LLMs, through a rational inattention (RI) perspective \cite{mackowiak2023rational}. Recent evidence shows that commonly used LLMs, including OpenAI GPT and Meta Llama, may exhibit bounded-rational behavior similar to humans \cite{macmillan2024ir}. A possible explanation is that the softmax policy employed in scaled dot-product transformer attention \cite{vaswani2017attention} has the same functional form as the optimal action policy in the RI framework (see Figure~\ref{fig:RI_LLMs}). This connection is leveraged as a principled mechanism for constructing an internal monitoring output, which indicates the rationality level of LLMs' outputs.

The contributions of this paper are threefold. First, we show that the attention weights used by LLMs can be interpreted as the solution to an RI-type optimization problem. This establishes a structural correspondence between transformer attention mechanisms and rational-inattention models of decision-making. 
Second, this connection enables us to introduce a head-wise bounded-rationality measure.
% Second, based on this connection, we introduce a head-wise bounded-rationality measure and characterize several of its basic properties, including positivity, monotonicity, and an explicit upper bound. 
Third, we use this measure as a monitoring output in a safety problem where the embedding inputs are under attack. This yields a divergence-based safety measure that quantifies the worst-case KL output-distribution shift induced by energy-bounded attacks that remain stealthy with respect to the bounded-rationality monitor. We also investigate perfectly undetectable attacks and derive an upper bound on the stealthy attack impact. The proposed measure is illustrated by a controlled text-prediction experiment on two pretrained causal language models: GPT-2 \cite{radford2019language} and GPT-Neo-125M \cite{black2021gpt}.

The remainder of this paper is organized as follows. Section~2 revisits the rational-inattention framework and transformer attention, and establishes their structural connection through a shared softmax-based optimization form. We introduce the head-wise bounded-rationality measure, which is used to formulate the proposed divergence-based safety measure for LLMs under embedding-input attacks in Section~3.
An analysis of the proposed safety measure is carried out in Section~4 by studying perfectly undetectable attacks and deriving an upper bound on the stealthy attack impact.
We present a controlled candidate-prediction experiment comparing GPT-2 and GPT-Neo-125M under the same attack budget and trust tolerance in Section~5, while Section~6 concludes the paper.

\begin{figure*}[!t]
    \centering
    \includegraphics[width=\linewidth]{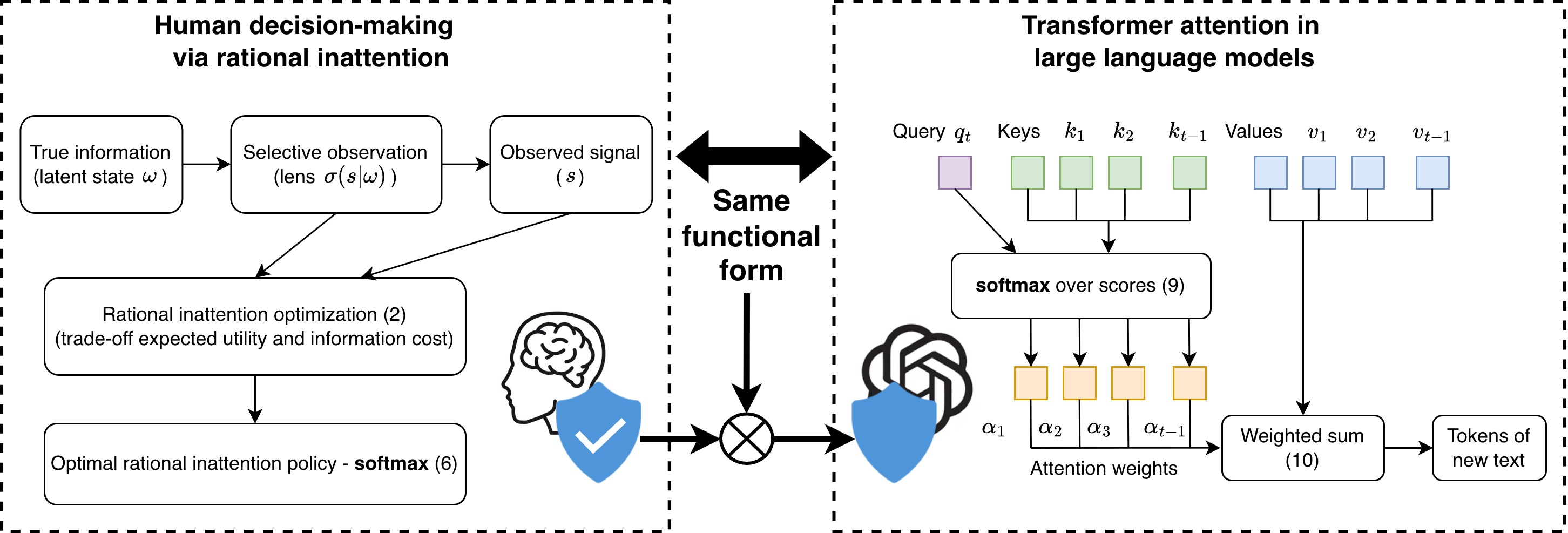}
    \caption{Illustration of the conceptual connection between RI in human decision-making and transformer attention in LLMs. In the RI framework, humans selectively allocate limited information-processing capacity to the most relevant observations before making decisions. Similarly, transformer attention mechanisms allocate attention weights to the most relevant context tokens when generating outputs. The figure highlights that both systems rely on the same softmax-based functional structure for selective information allocation, providing a bounded-rational interpretation of transformer attention in LLMs. This connection also enables us to leverage well-developed output perturbations for RI to quantify safety measures for LLMs.}
    % \vspace{-8pt}
    \label{fig:RI_LLMs}
\end{figure*}

\section{Rational Inattention and Large Language Models}
In this section, we revisit the RI framework and the transformer attention commonly used in LLMs. The purpose is to establish a mathematical connection between them.
\subsection{Rational inattention framework}

RI is a framework that formulates the human decision-making process when humans have a finite capacity for processing all the available data. In this framework, humans only observe a version of the true information via their customized lens, which can enhance the quality of the observed information by paying more attention at a higher cost. The better-quality data is expected to yield a higher utility for the human decision-maker. Therefore, the human decision-maker is interested in trading off the incremental utility and the extra information cost, which is the core formulation in the RI framework.

We revisit the RI framework by denoting a latent state $\omega \in \Omega$, which is not directly observed and has a probability $\pi(\omega)$. On the other hand, humans can only observe a state $s \in S$, which has a policy  
\begin{align}
    \sigma:\Omega \to \Delta(S).
\end{align}
Here, $\sigma(s|\omega)$ tells us the probability to observe $s$ when $\omega$ occurs. After observing $s$, we have an action policy $\mu(a|s)$ over an action space $\Ac$, i.e.,
\begin{align}
    \mu:S \to \Delta(\Ac).
\end{align}
With the action $a$ and the latent state $\omega$, we obtain the corresponding utility $u(a,\omega)$.

As discussed above, the information cost is formulated from the difference between observing $s$ and the latent state $\omega$ using the standard KL divergence as follows:
\begin{align}
    I(\omega ; s) &\triangleq D_{\text{KL}}(\pi(\omega) \sigma(s|\omega) \| \pi(\omega) \sigma(s)) \non \\
    &=  \sum_{\omega \in \Omega, \, s \in \Sc} \pi(\omega) \sigma(s|\omega) \log \frac{\sigma(s | \omega)}{\sigma(s)}. \label{def_inf_cost}
\end{align}
This information cost $I(\omega ; s)$ informs how well the state $s$ is observed. This can also be interpreted as how much attention is paid to $s$ when humans take an action, which is determined by solving the following optimization problem:
% Based on the above notation, an action policy is determined by solving the following optimization problem:
\begin{align}
    \sup_{\sigma, \mu} ~ \sum_{\omega \in \Omega, s \in \Sc, a \in \Ac} ~ \pi(\omega) \sigma(s|\omega) \mu(a|s) u(a,\omega) - \kappa I(\omega;s).
    \label{def_obj}
\end{align}
Here, the first term stands for the expected utility and the information cost is shown in the second term with $\kappa>0$, which plays a key role in the trade-off between the utility and the information cost.
%, which will appear in the optimal solution later. 
% In principle, the objective function \eqref{def_obj} aims to trade off between maximizing the expected utility and minimizing the information cost. 
We will discuss a solution to \eqref{def_obj} in the following.

Let us denote the latent-feedback action policy and the corresponding information cost in the same form as \eqref{def_inf_cost}: 
\begin{align}
    P(a|\omega) &\triangleq  \sum_{s \in \Sc} \mu(a|s) \sigma(s|\omega), \label{def_P_a_omega}
    \\
    I(\omega ; a) &\triangleq D_{\text{KL}}(\pi(\omega) P(a|\omega) \| \pi(\omega) P(a)) \non \\
    &=  \sum_{\omega \in \Omega, a \in \Ac} \pi(\omega) P(a|\omega) \log \frac{P(a|\omega)}{P(a)}. \label{def_info_cost_ri}
\end{align}
Then, by substituting \eqref{def_P_a_omega}--\eqref{def_info_cost_ri} into the optimization problem \eqref{def_obj}, we obtain the reduced optimization problem:
\begin{align}
    \sup_{P(a|\omega) \geq 0}& ~~ \sum_{\omega \in \Omega, a \in \Ac} \pi(\omega) P(a|\omega) u(a,\omega) - \kappa I(\omega ; a)~~ \label{def_obj_red}\\
    \text{s.t.}~~~&~~~~~  \sum_{a \in \Ac} P(a|\omega) = 1,~ \forall \omega \in \Omega. \non 
\end{align}

In the RI framework, the human decision-making process essentially amounts to solving \eqref{def_obj_red}, whose solution is presented in the following proposition.

\begin{Proposition}
\label{prop:RI_unique_sol}
        The optimization problem \eqref{def_obj_red} has a unique solution, which is
	 	\begin{align}
	 		P^\star(a|\omega) = \frac{P(a)  \exp \big(  u(a,\omega)/\kappa)}{\sum_{a' \in \Ac} P(a')  \exp \big(  u(a',\omega)/\kappa)}. 
	 		\label{ri_def_action}
	 	\end{align}
        \vspace{1pt}
\end{Proposition}

%\begin{proof}
%     \tcr{check with literature. How about the uniqueness?}
% \end{proof}

Proposition 1 follows from the generalized multinomial-logit characterization of the optimal RI choice rule under Shannon mutual-information cost (see \cite{matvejka2015rational,matejka2014supplementary} for more details).
In the following, we revisit the transformer attention \cite{vaswani2017attention}, commonly used in LLMs, with the aim of establishing its connection to the optimal action policy provided by the RI framework in \eqref{ri_def_action}.

\subsection{Large language models}
\label{sec:llm}

In short, LLMs are systems that can generate human-like text with the aim of reasonably responding to human queries. During the generation process, LLMs are provided with a sequence of text and are required to generate more text reasonably. In practice, LLMs do not directly work on the true text; instead, they work on the embedding tokens (vectors) of the true text, which are denoted as
\begin{align}
	\{x_i\}_{1\le i\le t-1} \triangleq \{x_1, x_2, \ldots, x_{t-1}\},~ x_i \in \Rbb^d, \label{llm_input_tokens}
\end{align}
at time $t$. The model is required to predict the next token $x_t$ based on the previous tokens $\{x_i\}_{1\le i\le t-1}$.

To predict the next token $x_t$ more precisely, the model uses a multi-head scheme, where each head distributes its own attention on the previous tokens and produces its own internal representation for generating $x_t$. Concatenating all the internal representations given by all the heads gives an overall internal representation that decides the next token $x_t$. This scheme is inspired by invoking multiple experts' views on the same problem to make the best decision.

Let us assume that the model employs a set of $M$ heads with indices \(\{1,2, \ldots, \allowbreak M\}\), where each head uses query, key, and value with the corresponding learned mappings
\begin{align}
	W_q^{(m)} \in \Rbb^{d_k \times d}, ~~
	W_k^{(m)} \in \Rbb^{d_k \times d}, ~~
	W_v^{(m)} \in \Rbb^{d_v \times d}, ~~\forall m \in \{1,2,\ldots,M\}. \non
\end{align} 
We assume that these learned mappings are given by a pre-trained phase, which are used to compute query, key, and value as follows:
\begin{align}
	q_t^{(m)}  &\triangleq W_q^{(m)}  x_t \in \Rbb^{d_k},~~
	k_i^{(m)}  \triangleq W_k^{(m)}  x_i \in \Rbb^{d_k}, \\ 
	v_i^{(m)}  &\triangleq W_v^{(m)}  x_i \in \Rbb^{d_v}, ~~ \forall i \leq t-1. \label{llm_def_qkv}
\end{align}

Let us define the following correlation score between the prediction token $x_t$ and the past tokens $\{x_i\}_{1 \leq i \leq t-1}$:
\begin{align}
	s_{t,i}^{(m)} \triangleq {q_t^{(m) \top} k_i^{(m)}}, ~~ \forall i \leq t-1. \label{llm_def_input}
\end{align}
Next, this correlation is used to compute the attention weight by a softmax function as follows:
\begin{align}
	\alpha^{(m)}_{t,i} \triangleq \frac{\exp({s^{(m)}_{t,i}/\tau)}}{\sum_{i=1}^{t-1} \exp({s^{(m)}_{t,i}/\tau})}, \label{llm_def_attention_weight}
\end{align}
where $\tau = {\sqrt{d_k}}$ is a temperature parameter. The use of this temperature parameter enables the model to work better on long sentences \cite{vaswani2017attention}. As a result, the internal representation of each head $m$ is formulated as follows:
\begin{align}
	h_t^{(m)} \triangleq  \sum_{i = 1}^{t-1} \alpha_{t,i}^{(m)} v^{(m)}_i. \label{llm_def_int_rep}
\end{align}
The combined internal representation given by multiple heads can be described as follows:
\begin{align}
	h_t \triangleq  \sum_{m=1}^M ~ W_O^{(m)} h^{(m)}_t,
    \label{def_multi_head_ht}
\end{align}
where $W_O^{(m)}$ is a given learned mapping for each head $m$.

Let us take a pause for observing the attention weight \eqref{llm_def_attention_weight}, which takes the same functional form as \eqref{ri_def_action}. This observation leads to the following claim.
\begin{Proposition}
\label{prop:llm_opt_prob}
		Consider the reduced RI objective function \eqref{def_obj_red} with its unique solution \eqref{ri_def_action} and the attention weight of LLMs \eqref{llm_def_attention_weight}. The attention weight \eqref{llm_def_attention_weight} is a unique solution to the following optimization problem:
		\begin{align}
			V^{(m)}(\tau) \triangleq  \sup_{\alpha^{(m)}_{t,i} \geq 0}&  ~~ \sum_{i=1}^{t-1} \alpha^{(m)}_{t,i} q^{(m) \top}_t k^{(m)}_i 
            - \tau  \sum_{i=1}^{t-1}\alpha^{(m)}_{t,i} \log (\alpha^{(m)}_{t,i}) \label{llm_def_opt} \\
			\text{s.t.}~~&~~  \sum_{i=1}^{t-1} \alpha^{(m)}_{t,i} = 1. 
		\end{align}
\end{Proposition}

It is worth noting that the second term of the objective function \eqref{llm_def_opt}, i.e., $H(\alpha^{(m)}) = -\sum_{} \alpha^{(m)}_{t,i} \log (\alpha^{(m)}_{t,i})$, is the so-called Shannon information entropy, which quantifies the uncertainty of the information. Substituting the optimal solution \eqref{llm_def_attention_weight}, which satisfies the constraint in \eqref{llm_def_opt}, into the objective in \eqref{llm_def_opt} yields
\begin{align}
	\bar V^{(m)}(\tau) = \tau \log  \sum_{i=1}^{t-1} \exp({q^{(m)\top}_t k^{(m)}_i/\tau}), \label{llm_def_V_bounded}
\end{align}
which we call a $\tau$-bounded rationality value.

\begin{remark}
    Proposition~\ref{prop:llm_opt_prob} points out that the transformer attention commonly used in LLMs is essentially the solution to an RI-type optimization problem, which has the same form as the objective in the RI framework. This result establishes a connection between the transformer attention in LLMs and human decision-making with bounded rationality, confirming the relationship illustrated in Figure~\ref{fig:RI_LLMs}. \QET
\end{remark}

It can be argued that the RI-based problem \eqref{llm_def_opt} is fully rational if it does not pay for information cost $H(\alpha^{(m)})$, by setting $\tau = 0$, resulting in the following linear program:
\begin{align}
	V^{(m)}_0 =  \sup_{\hat \alpha^{(m)}_{t,i} \geq 0}&~~  \sum_{i=1}^{t-1} \hat \alpha^{(m)}_{t,i} q^{(m)\top}_t k^{(m)}_i  \label{llm_def_opt_full} \\
	\text{s.t.}~~&~~  \sum_{i=1}^{t-1} \hat \alpha^{(m)}_{t,i} = 1.  \non
\end{align}
Simply, solving \eqref{llm_def_opt_full} gives us a (non-unique) solution: $\hat \alpha^{(m)\star}_{t,i} = 1$ if $q^{(m)\top}_t k^{(m)}_i > q^{(m)\top}_t k^{(m)}_j ~\forall j, j \neq i, 1 \leq j \leq t-1$ and $\hat \alpha^{(m)\star}_{t,i} = 0$ otherwise,
yielding the corresponding optimal value of \eqref{llm_def_opt_full} as:
\begin{align}
	\bar V^{(m)}_0 =  \max_{i \in \{1,2,\ldots,t-1\}} ~~ q^{(m)\top}_t k^{(m)}_i, 
	\label{llm_def_V_full}
\end{align}
which is referred to as the full rationality value.

\section{Performance Measures}

In this section, we aim to leverage \eqref{llm_def_V_bounded} and \eqref{llm_def_V_full}, introduced in the previous section, to build a rationality measure and show how it can be used to develop a safety measure.

\subsection{Bounded rationality measure}
From \eqref{llm_def_V_bounded} and \eqref{llm_def_V_full}, we define the difference between these two values as a $\tau$-bounded rationality measure for each head $m$ as follows:
\begin{align}
	\Delta V^{(m)}(\tau) &\triangleq \bar V^{(m)}(\tau) - \bar V^{(m)}_0.  
    %&=\tau \log \sum_{i=1}^{t-1} \exp({q^{(m)\top}_t k^{(m)}_i/\tau})- \max_{i \in \{1,2,\ldots,t-1\}} ~ q^{(m)}_t k^{(m)}_i. 
    \label{llm_def_br_metric}
\end{align}
This measure has some properties listed in the following.
\begin{Proposition}
\label{prop:br_metric_properties}
        Suppose the history tokens \eqref{llm_input_tokens} have a length of at least two, i.e., $t > 2$.
		Consider the $\tau$-bounded rationality measure \eqref{llm_def_br_metric}. It has the following four properties:
		\begin{enumerate}
			\item $\lim_{\tau \ra 0^+} \Delta V^{(m)}(\tau) = 0$;
			\item $\Delta V^{(m)}(\tau) > 0~~  \forall \tau > 0$;
			\item $\Delta V^{(m)}(\tau)$ is strictly increasing in $\tau$;
			\item $\Delta V^{(m)}(\tau) \leq \tau \log(t-1)$. \QET
		\end{enumerate}
\end{Proposition}

\textbf{Proof:}
Since the first property is obtained by the structure of \eqref{llm_def_opt_full} built from \eqref{llm_def_opt}, it suffices to show $\displaystyle \frac{\partial}{\partial \tau} \Delta V^{(m)}(\tau) > 0~\forall \tau$, which is straightforward, to obtain the second and third properties. 
% Regarding the last property, we denote
% $s^{(m)\star}_t \triangleq \max_{i \in \{1,2,\ldots,t-1\}} ~ q^{(m)\top}_t k^{(m)}_i$, which leads to
%%
The last property is shown through the following relations:
\begin{align}
	 \sum_{i=1}^{t-1} \exp({q^{(m)\top}_t k^{(m)}_i/\tau})& \leq  (t-1) \exp({\bar V^{(m)}_0/\tau}) \non \\
	 \log \sum_{i=1}^{t-1} \exp({q^{(m)\top}_t k^{(m)}_i/\tau}) & \leq 
	\log(t-1) + \frac{\bar V^{(m)}_0}{\tau} \non \\
	\Delta V^{(m)}(\tau)  &\leq \tau \log(t-1). 
\end{align}
Here, the first and last inequalities hold by \eqref{llm_def_V_full}-\eqref{llm_def_br_metric}.
\QEDB

\begin{remark}
\label{rem:use_case_br_metric}
    From the $\tau$-bounded rationality measure in \eqref{llm_def_br_metric} and its properties in Proposition~\ref{prop:br_metric_properties}, this measure should be close to zero when the response is close to full rationality. Since the multi-head scheme is employed, each head should put its main focus on one particular portion of the embedding inputs rather than diffusing its focus over all the embedding inputs. This scheme is motivated by real-world complex problem-solving practices. Take a group of cross-disciplinary experts working on a single project as an example \cite{ERC_Synergy_Grant}, where each expert works on a particular part of the project to provide a wide-angle view of the problem. \QET
\end{remark}

The next question is how the bounded-rationality measure \eqref{llm_def_br_metric} can be used in practice. We foresee that, in the next few years, LLMs will be deployed in many real-world applications, where human operators need to consult them before making crucial decisions. Besides the response itself, human operators desire to know how much bounded rationality the response has. This information would ease the uncertainty in the trust between humans and LLMs. For example, human operators can monitor the $\tau$-bounded rationality measure at every LLM response before making a decision. If the rationality measure over time remains within a small range near zero, the responses might correspond to high rationality. Otherwise, 
%a low rationality, i.e., a high value of $\Delta V^{(m)}(\tau)$, should raise a warning to the human operator.
small values of $\Delta V^{(m)}(\tau)$ correspond to more concentrated attention allocations, which reflect reasoning quality or rationality and warrant further empirical investigation.
This practice is in line with fault-detection techniques in control systems \cite{ding2008model}. Leveraging this bounded-rationality measure, we study a safety problem in LLMs when the embedding input tokens in \eqref{llm_input_tokens} are under attack in the following.

\subsection{Divergence-based safety measure}
In this subsection, we study a safety problem in the responses provided by LLMs by assuming that the embedding tokens \eqref{llm_input_tokens} are maliciously manipulated by an adversary. Let us denote the attacked embedding tokens as:
\begin{align}
    \{ x_1^a, x_2^a, \ldots, x_{t-1}^a \}, \label{llm_input_token_attacked}
\end{align}
where $x_i^a \triangleq x_i + \zeta_i$ and $\zeta_i \in \Rbb^d$ is the attack signal injected into token $x_i$. We assume that the attack signals $\{\zeta_i\}_{1 \leq i \leq t-1}$ are under an energy constraint, i.e.,
\begin{align}
     \sum_{i=1}^{t-1}~ \norm{\zeta_i}^2_2 \leq \rho, \label{llm_energy_bound}
\end{align}
where $\rho$ is a given maximum attack energy.

The aim of the manipulation is to falsify the predicted token $x_t$, resulting in a bad response given by LLMs to human operators. Here, we use the $\tau$-bounded rationality measure \eqref{llm_def_br_metric} as a monitoring indicator for this safety problem, where we want to keep the weighted sum of this measure for all the heads under a given threshold $\epsilon_T$, i.e., 
\begin{align}
    \sum_{m = 1}^M w^{(m)} \Delta V^{(m)}(\tau) \leq \epsilon_T.
\end{align}
The non-negative per-head rationality weight $w^{(m)}$ is given and satisfies
\begin{align}
    \sum_{m=1}^M w^{(m)} = 1.
\end{align}
We can call $\epsilon_T$ a trust tolerance. Human operators should raise a concern about LLMs' responses when this trust tolerance is exceeded.
In Section~\ref{sec:llm}, we provided the computation for the internal representation $h_t$ in \eqref{def_multi_head_ht}, which is a crucial variable to predict the next token through the probability over the action space $\Ac$ with the following softmax law:
\begin{align}
    P(a|h_t) \triangleq \frac{\exp{(z(a,h_t)/\kappa)}}{\sum_{a' \in \Ac} \exp{(z(a',h_t)/\kappa)}}, 
    \label{def_P_a_ht}
\end{align}
where $\kappa > 0$ is a given temperature parameter and 
\begin{align}
    z(\cdot,h_t): \Rbb^{d_h} \ra \Rbb^{|\Ac|} 
    %\triangleq  \sum_{\omega \in \pi} \pi(\omega|h_t) [ u_{LLM}(a, \omega, h_t) | h_t ]. 
    \label{def_z_ah}
\end{align}
as the output score map of the LLM.

~\\
On the other hand, we have the action probability corresponding to the attacked tokens \eqref{llm_input_token_attacked} as $P(a|h_t^a)$. Inspired by the output-to-output gain security measure in \cite[Sec. II.B]{nguyen2025scalable}, which quantifies the worst-case impact of stealthy attacks, we introduce the following divergence-based safety measure for one-step prediction:
\begin{align}
     \sup_{\{\zeta_i\}_{1\leq i \leq t-1}}&~~D_{\text{KL}}(P(a|h^a_t) || P(a|h_t)) \label{llm_sec_metric_one}\\
    \text{s.t.}~~~~&~~ \sum_{m =1}^{M} w^{(m)}   \Delta V^{(m)}(\tau)  \leq \epsilon_T, \non \\
    &~~  \sum_{i=1}^{t-1} \norm{\zeta_i}_2^2 \leq \rho. \non 
\end{align}
Here, $D_{\text{KL}}(P(a|h_t^a) || P(a|h_t))$ stands for the KL divergence of $P(a|h_t^a)$ from nominal action $P(a|h_t)$. The first constraint in \eqref{llm_sec_metric_one} is the stealthiness condition, where the adversary wants to be stealthy to the human operator monitoring the rationality level by keeping it under the trust tolerance $\epsilon_T$. The last constraint in \eqref{llm_sec_metric_one} is identical to \eqref{llm_energy_bound}.

% \begin{remark}
%     It is more interesting to investigate \eqref{llm_sec_metric_one} when only one input token is falsified, i.e., only one single word is manipulated. A large change in the predicted token with a small change in the bounded rationality level is potentially harmful to the human operator receiving the manipulated response. Clearly, this special case is captured in \eqref{llm_sec_metric_one}.
%     \QET
% \end{remark} 

\section{Safety Measure Analysis}
In this section, we analyze the proposed safety measure \eqref{llm_sec_metric_one}. We first discuss perfectly undetectable attacks and then characterize an upper bound for the safety measure \eqref{llm_sec_metric_one}.

\subsection{Perfectly undetectable attacks}
% Let us assume that the embedding input tokens without attacks provide a monitoring indicator below the trust tolerance $\epsilon_T$.
Consider the safety measure \eqref{llm_sec_metric_one}. We are interested in providing a condition under which LLMs are insecure with respect to \eqref{llm_sec_metric_one}, which is presented in the following. 
\begin{Definition}[Insecure LLMs]
    \label{def:insecure_llms}
    Suppose \eqref{llm_sec_metric_one} is the safety measure for an LLM. The LLM is called insecure if the following condition holds:
    \begin{align}
         \sum_{i=1}^{t-1} \alpha^{(m)}_i W^{(m)}_v \hat \zeta_i \neq 0~ \text{for some}~m, \label{pua_impact}
    \end{align}
    where $\{\hat \zeta_i\}_{1 \leq i \leq t-1} \in \bigcap_{m=1}^M ~ \Nc(W^{(m)}_k)$. \QET
    \vspace{5pt}
\end{Definition}

From the safety measure \eqref{llm_sec_metric_one}, we are interested in attack signals $\{\zeta_i\}_{1 \leq i \leq t-1}$, such that they give the same monitoring indicator as the attack-free case. This type of attack is called a perfectly undetectable attack \cite{milovsevic2020actuator}.
Based on the definitions \eqref{llm_def_qkv}, \eqref{llm_def_br_metric}, and \eqref{llm_input_token_attacked}, the perfectly undetectable attack satisfies the following sufficient condition:
\begin{align}
    q^{(m)\top}_t W^{(m)}_k \zeta_i = 0,~~ &\forall i \in \{1,2,\ldots, t-1\}, \label{llm_perfect_stealthy_attack} \\
    &\forall m \in \{1,2,\ldots, M\},~ \zeta_i \neq 0. 
\end{align}

Since $\{\zeta_i\}_{1\le i\le t-1}$ is a solution to \eqref{llm_perfect_stealthy_attack}, the scaled family $\{c \zeta_i\}_{1\le i\le t-1}$ is also a solution for any $c \neq 0$. As a consequence, if we relax the energy constraint in \eqref{llm_sec_metric_one}, the safety measure can be made arbitrarily large. Clearly, the sufficient condition \eqref{llm_perfect_stealthy_attack} yields a nontrivial solution when the key mappings of all the heads share a common nullspace \eqref{pua_impact}, i.e., $\bigcap_{m=1}^M ~ \Nc(W^{(m)}_k) \neq \{ 0 \}$. 
% \begin{align}
%     \bigcap_{m=1}^M ~ \Nc(W^{(m)}_k) \neq \emptyset. \label{pua_cond}
% \end{align}

The impact of a perfectly undetectable attack is serious when it causes a change in the internal representation \eqref{llm_def_int_rep} in some heads \eqref{pua_impact}, resulting in $h^{a(m)}_t - h^{(m)}_t \neq 0$ for some $m$. From the discussion above, this change in the internal representation $h_t$ can be enlarged by changing $c$ arbitrarily.  

\begin{remark}
    Definition~\ref{def:insecure_llms} provides a condition to enhance the safety for LLMs. In practice, one can check the triviality of $\bigcap_{m=1}^{M}\mathcal N\!\bigl(W_k^{(m)}\bigr)$. If it is not trivial, one may slightly perturb some matrices $W_k^{(m)}$ with a very small performance loss so as to avoid perfectly undetectable attacks. In a more advanced technique, the triviality condition can be incorporated into the learning phase of the matrices $W_k^{(m)}$, inspired by physics-informed machine learning techniques.  \QET
\end{remark}

\subsection{Upper bound of the safety measure}
In this part, we aim to characterize the upper bound of the safety measure \eqref{llm_sec_metric_one}. 
The characterization is presented after we introduce a mild assumption in the following.
\begin{Assumption}
    \label{assumption:lipschitz}
    The mapping \eqref{def_z_ah} is Lipschitz continuous in $h$, i.e., there exists a positive constant $L_z > 0$ such that
    \begin{align}
        \norm{z(\cdot,h_t) - z(\cdot,h_t^a)}_2 \leq L_z \norm{h_t - h_t^a}_2, ~\forall h_t, h_t^a. \label{ht_lipzchits}
    \end{align}
\end{Assumption}
\begin{Proposition}
    \label{prop:DKL_upperbound}
    Suppose Assumption~\ref{assumption:lipschitz} holds. The objective function in \eqref{llm_sec_metric_one} is upper bounded:
    \begin{align}
        D_{\text{KL}}(P(a|h_t^a) || P(a|h_t)) \leq  \gamma \norm{h_t - h_t^a}_2^2, \label{DKL_upperbound}
    \end{align}
    where $\gamma = \frac{L_z^2}{4 \kappa^2}$ with $L_z$ being the Lipschitz constant in \eqref{ht_lipzchits} and $\kappa$ being a temperature parameter in \eqref{def_P_a_ht}. \QET
\end{Proposition}
\begin{proof}
    Let us use some auxiliary variables
    \begin{align}
        \hat z \triangleq \frac{z(\cdot,h_t)}{\kappa}, ~~
        \hat z^a \triangleq \frac{z(\cdot,h_t^a)}{\kappa},  \label{def_hat_z}
    \end{align}
    where $z(\cdot,h)\in\mathbb R^{|\mathcal A|}$ collects the values
    $\{z(a_i,h)\}_{a_i\in\mathcal A}$ defined in \eqref{def_z_ah}. Let us denote
    \begin{align}
        \phi(\zeta) \triangleq  \log \sum_{i = 1}^{|\Ac|} \exp(\zeta_i), ~~ \zeta \in \Rbb^{|\Ac|},  
    \end{align}
    which has
    %$\nabla \phi(\zeta) = \text{softmax}(\zeta)$. 
    \begin{align}
        \nabla \phi(\zeta) = \frac{\exp(\zeta_i)}{\sum_{i = 1}^{|\Ac|} \exp(\zeta_i)} \triangleq \text{softmax}(\zeta).
    \end{align}
    This new notation assists us in re-describing action policies \eqref{def_P_a_ht} as:
    \begin{align}
        P(a|h_t) &= \nabla \phi(\hat z),
        ~P(a|h_t^a) = \nabla \phi(\hat z^a), \label{action_vector_new_form} \\
         \log P(a_i|h_t) &= \hat z_i - \phi(\hat z),~
        \log P(a_i|h_t^a) = \hat z_i^a - \phi(\hat z^a). \label{action_new_form}
    \end{align}

    By definition, the KL divergence in the left-hand side of \eqref{DKL_upperbound} has the following equivalent form:
    \begin{align}
        &D_{\text{KL}}(P(a|h_t^a) || P(a|h_t)) =  \sum_{i=1}^{|\Ac|} P(a_i|h_t^a) \log \frac{P(a_i|h_t^a)}{P(a_i|h_t)} \\
        =&~  \sum_{i = 1}^{|\Ac|} P(a_i|h_t^a) \big( \log P(a_i|h_t^a) -   \log P(a_i|h_t) \big) \\
        =&~  \sum_{i = 1}^{|\Ac|} P(a_i|h_t^a) (\hat z_i^a - \phi(\hat z^a) - \hat z_i + \phi(\hat z)) \\
        =&~ \phi(\hat z) - \phi(\hat z^a) + P(a|h_t^a)^\top(\hat z^a - \hat z)
        \\
        =&~\phi(\hat z) - \phi(\hat z^a) - \nabla \phi(\hat z^a)^\top(\hat z - \hat z^a),
        \label{KLD_new_form}
    \end{align}
    where the first two equalities are given by definition, the third equality comes from \eqref{action_new_form}, the fourth equality comes from the fact that $\phi(\hat z^a)$ and $\phi(\hat z)$ are independent of $P(a_i|h_t^a)$ with $\sum_{i=1}^{|\Ac|}P(a_i|h_t^a) = 1$, and the last equality comes from \eqref{action_vector_new_form}. Now, the right-hand side of \eqref{KLD_new_form} is exactly the Bregman divergence associated with $\phi(\cdot)$. The integral form of the Bregman divergence gives us the following equality:
    \begin{align}
        &D_{\text{KL}}(P(a|h_t^a) || P(a|h_t))  \\
        =& \int_{0}^1 (1-s)(\hat z - \hat z^a)^\top \nabla^2 \phi(\hat z^a + s(\hat z - \hat z^a)) (\hat z - \hat z^a) \text{d}s \\
        \leq&\frac{1}{2} \norm{\hat z - \hat z^a}^2_2 \int_{0}^1 (1-s)\text{d}s   
        ~=  \frac{1}{4} \norm{\hat z - \hat z^a}_2^2 \\
        =&~ \frac{1}{4 \kappa^2} \norm{z(\cdot,h_t) - z(\cdot,h_t^a)}^2_2, \label{DKL_ineq_1/4}
    \end{align}
    where the last equality comes from the definitions \eqref{def_hat_z} and the first inequality comes from the following inequality (see Appendix A for more details): 
    \begin{align}
        \nabla^2 \phi(\zeta) = \text{diag} (\text{softmax}(\zeta)) - \text{softmax}(\zeta) \text{softmax}(\zeta)^\top \preceq \frac{1}{2} I.
    \end{align}
    % \begin{align}
    %     \nabla^2 \phi(\zeta) = \text{diag} (\text{softmax}(\zeta)) - \text{softmax}(\zeta) \text{softmax}(\zeta)^\top \preceq \frac{1}{2} I.
    % \end{align}
    Finally, \eqref{DKL_ineq_1/4} is combined with \eqref{ht_lipzchits} to result in \eqref{DKL_upperbound}. This completes the proof. \QEDB
\end{proof}

The result of Proposition~\ref{prop:DKL_upperbound} allows us to consider the following alternative optimization problem, whose solution, multiplied by $\gamma$, serves as an upper bound for \eqref{llm_sec_metric_one}:
\begin{align}
     \sup_{\{\zeta_i\}_{1\leq i \leq t-1}}&~~\norm{h_t - h_t^a}_2^2 \label{llm_sec_metric_one_alt}\\
    \text{s.t.}~~~&~~ \sum_{m =1}^{M} w^{(m)}   \Delta V^{(m)}(\tau)  \leq \epsilon_T, \non \\
    &~~ \sum_{i=1}^{t-1} \norm{\zeta_i}_2^2 \leq \rho. \non 
\end{align}
It is worth noting that \eqref{llm_sec_metric_one_alt} can be computed using sequential quadratic programming \cite{boggs1995sequential} when we employ an approximation for $\bar V_0^{(m)}$ in \eqref{llm_def_br_metric}, which is $\bar V_0^{(m)} \approx \bar V^{(m)}(\eta)$ with $\eta$ is a very small positive value thanks to properties in Proposition~\ref{prop:br_metric_properties}.

In the special case of perfectly undetectable attacks, discussed in the previous subsection, this type of attack naturally fulfills the first condition in \eqref{llm_sec_metric_one_alt}. Given that $h_t - h_t^a = - \sum_{m=1}^M W_O^{(m)} \sum_{i=1}^{t-1} \alpha_{t,i}^{(m)} W^{(m)}_v \zeta_i$, \eqref{llm_sec_metric_one_alt} is reduced to the following form with its closed-form solution:
\begin{align}
     &\sup_{\{\zeta_i\}_{1\leq i \leq t-1}}~ \norm{ \sum_{i=1}^{t-1} \sum_{m=1}^M W_O^{(m)}  \alpha_{t,i}^{(m)} W^{(m)}_v \zeta_i}_2^2 \label{llm_sec_metric_one_alt_pua} \\
    &~~~~~\text{s.t.}~~~~~~
     \sum_{i=1}^{t-1} \norm{\zeta_i}_2^2 \leq \rho. \non \\
    &=~  \rho \lambda_{\max}  \big( (I_{t-1} \otimes N_k)^\top B^\top B (I_{t-1} \otimes N_k) \big),
\end{align}
where $I_{t-1}$ is the $(t-1)$-dimensional identity matrix, $N_k$ is defined such that $W^{(m)}_k N_k = 0$ for all $m$, and $B  \triangleq \big[ B_1, B_2, \ldots, B_{t-1} \big]$ with 
\begin{align}
   B_i = \sum_{m=1}^M W_O^{(m)}  \alpha_{t,i}^{(m)} W^{(m)}_v 
\end{align}
for all $1 \leq i \leq t-1$. Here, $\lambda_{\max}(\cdot)$ stands for the maximum eigenvalue and $\otimes$ is the Kronecker product.

\section{Experimental Results}
\label{sec:experiments}

\begin{figure}[!t]
    \centering
    \includegraphics[width=\linewidth]{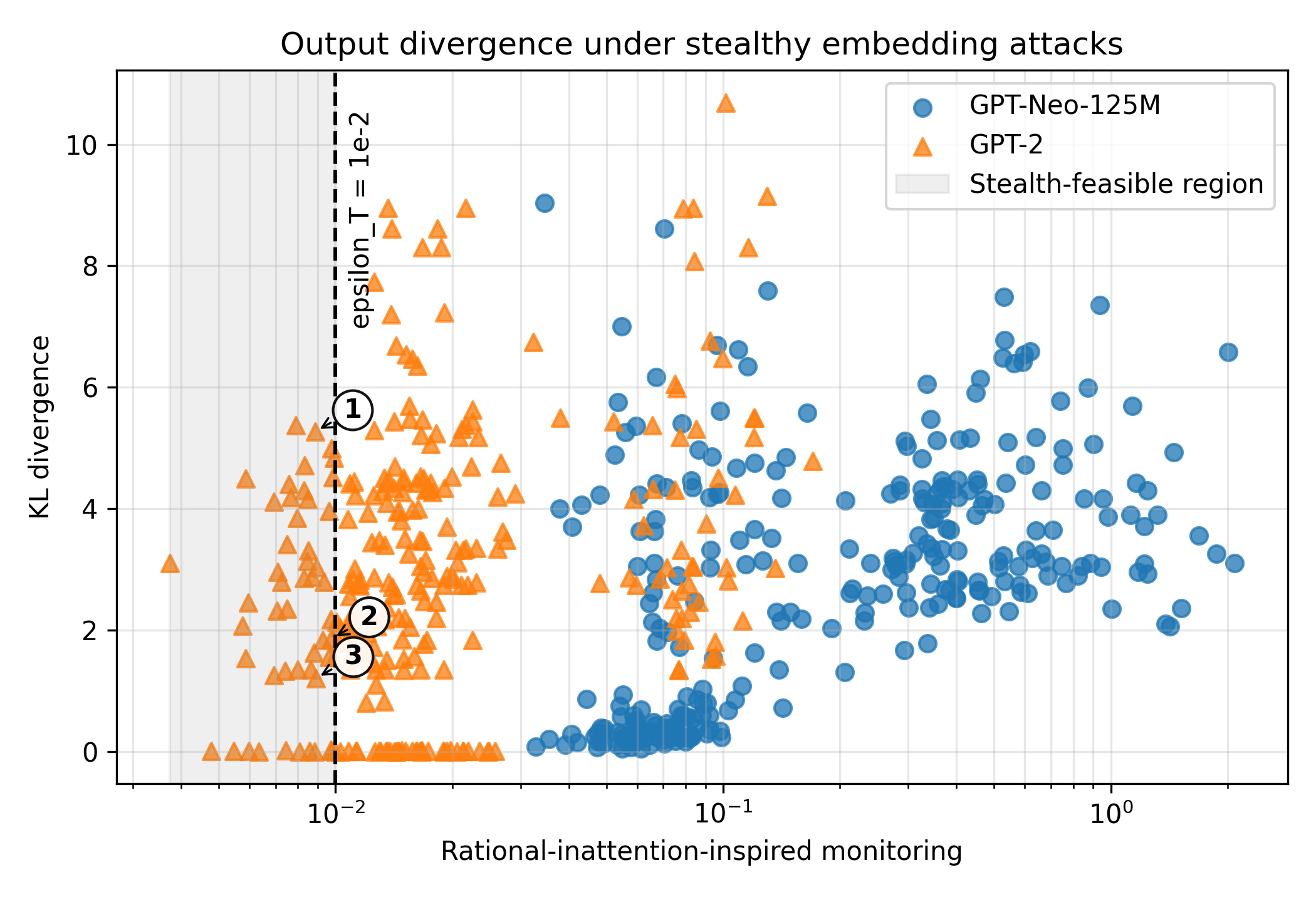}
    % \vspace{-20pt}
    \caption{Two-model comparison under the proposed divergence-based safety measure.
    Each point corresponds to one embedding-input attack on the candidate-adjective prediction task.
    Larger KL values inside this region indicate larger stealthy output changes under the same attack budget.
    The numbered GPT-2 examples correspond to the following prediction changes:
    (1) gate: open $\rightarrow$ unlocked;
    (2) laboratory: open $\rightarrow$ closed;
    and (3) restaurant: open $\rightarrow$ closed.
    }
    % \vspace{-10pt}
    \label{fig:two_model_safety_comparison}
\end{figure}

We illustrate the proposed divergence-based safety measure \eqref{llm_sec_metric_one} on two pretrained causal language models: GPT-2 \cite{radford2019language} and GPT-Neo-125M \cite{black2021gpt}. The purpose of the experiment is to show that the proposed measure can compare model-level safety under the same attack budget and trust tolerance. The code used to obtain the results is available at
\url{https://github.com/tungnguyenkstnk58/RI_LLMs}. The implementation was developed with partial assistance from Microsoft Copilot, accessed through \href{https://www.uu.se/en/students/it-for-students/software/ai-tool}{\textit{Uppsala University's licensed version}} under its Microsoft data-protection agreement, and verified by the authors.

Both models used in the experiment have $L=12$ transformer layers and $H=12$ attention heads per layer. Therefore, the multi-head monitoring scheme contains $LH=144$ layer-head pairs. The hidden dimension is $768$, and each head has dimension $d_k=64$, giving the attention temperature in \eqref{llm_def_attention_weight} as $\tau=\sqrt{d_k}=8$.

We use the same controlled candidate-prediction task for both models. The prompt set consists of ambiguous sentences of the form ``The [subject] [verb] [preposition] the [object] because it is''. The ambiguity arises from the pronoun ``it'', whose antecedent may be either the [subject] or the [object], resulting in different plausible adjective predictions. The next-token distribution is restricted to a shared pool of common predicative physical/state adjectives, e.g., \texttt{open}, \texttt{closed}, \texttt{hard}, \texttt{soft}, \texttt{safe}, \texttt{dangerous}, \texttt{hot}, and \texttt{cold}. For each prompt, all input-token embeddings are perturbed under the same energy budget, and attacks are generated by projected gradient ascent using the candidate-adjective KL divergence and the rational-inattention-inspired monitoring residual.

The experimental results are reported in 
Figure~\ref{fig:two_model_safety_comparison}, where we compare the attack-induced output divergences of GPT-2 and GPT-Neo-125M. Each point corresponds to one successful embedding-input attack, which maximally alters the output prediction. 
%The horizontal axis reports the rational-inattention-inspired monitoring residual, and the vertical axis reports the KL divergence between the attacked and clean candidate-adjective distributions. 
The shaded region denotes the trust tolerance $\epsilon_T = 10^{-2}$. Therefore, points inside this region correspond to attacks that remain stealthy with respect to the monitoring output. Larger KL values inside this region indicate larger stealthy output changes under the same attack budget.

The two models exhibit different safety profiles. GPT-2 admits several high-divergence attacks close to or inside the stealth-feasible region, whereas many high-impact GPT-Neo-125M attacks occur at larger monitoring residuals. This illustrates that the proposed divergence-based safety measure can distinguish the model-level safety under the same attack budget and trust tolerance.

% \vspace{-1pt}
\section{Conclusions}
This paper introduced a divergence-based safety measure for large language models under embedding-input attacks. The measure quantifies the worst-case KL divergence between clean and attacked output distributions subject to attack-energy and stealthiness constraints. The stealthiness monitor was constructed from a rational-inattention perspective by interpreting transformer attention as an entropy-regularized information-allocation mechanism, which led to a head-wise bounded-rationality measure. Experiments on pretrained language models showed that the proposed measure can distinguish model-level safety profiles under the same attack budget and trust tolerance. This work is expected to contribute to the development of interpretable performance measures for LLMs in the same spirit that control systems are equipped with robustness, resilience, and safety measures.

\bibliographystyle{splncs04}
\bibliography{mybibfile}

\appendix
\section{}
\begin{Lemma}
    For any vector $p \in \Rbb^{q}$, where $\sum_i^q p_i = 1$ and $p_i \geq 0$, we have 
\begin{align}
        \text{diag}(p_t) - p_t p_t^\top \preceq \frac{1}{2} I. \label{pt_cond_1/2}
    \end{align}
    \QET
\end{Lemma}
\begin{proof}
    Let us denote $M(p) = \text{diag}(p) - p p^\top$, where $[M(p)]_{ii} = p_i - p_i^2$ and $[M(p)]_{ij} = -p_i p_j$ for $j \neq i$. Using the Gershgorin circle theorem, the eigenvalue of $M(p)$, denoted $\lambda_M$, lies in the following range for some $i$:
    \begin{align}
        \textstyle \lambda_M &\geq  |[M(p)]_{ii}| - \textstyle \sum_{j\neq i} |[M(p)]_{ij}|, \\
        \lambda_M &\leq  |[M(p)]_{ii}| + \textstyle \sum_{j\neq i} |[M(p)]_{ij}|.
    \end{align}
    On the other hand, $p_i + \sum_{j \neq i} p_j = 1$ gives us
    \begin{align}
        \textstyle \sum_{j\neq i} |[M(p)]_{ij}| = p_i \textstyle \sum_{j \neq i} p_j 
        = p_i (1 - p_i) = [M(p)]_{ii}.
    \end{align}
    Therefore, the eigenvalue $\lambda_M$ lies in the range $\big[0, 2p_i (1 - p_i) \big]$. Furthermore, we also have 
   \begin{align}
       M(p) \preceq \max \lambda_M I \preceq 2p_i (1 - p_i) I = \frac{1}{2} I - \bigg( \sqrt{2}p_i - \frac{1}{\sqrt{2}} \bigg)^2 I \preceq \frac{1}{2} I.
   \end{align}
   \QEDB
\end{proof}

\end{document}